%% file: Main.tex
\documentclass[11pt]{article}

\usepackage[numbers]{natbib}
\usepackage{amsmath}
\usepackage[T1]{fontenc}
\usepackage{geometry}
\usepackage{amsthm}
\usepackage{amssymb}
\usepackage{graphicx}
\usepackage{esint}

\newtheorem{theorem}{Theorem}[section]
\newtheorem{definition}[theorem]{Definition}
\newtheorem{lemma}[theorem]{Lemma}
\newtheorem{corollary}[theorem]{Corollary}
\theoremstyle{definition}
\newtheorem{example}{Example}[section]

\newcommand{\vpc}{v}				
\newcommand{\vpcprof}{\mathbf{\vpc}}	
\newcommand{\ctr}{p}				
\newcommand{\sctr}{s}			
\newcommand{\m}{m}				
\newcommand{\n}{n}				
\newcommand{\val}{v}				
\newcommand{\reval}{r}			
\newcommand{\revalvec}{\mathbf{\reval}}			

\newcommand{\qt}{t}						
\newcommand{\QT}{T}						
\newcommand{\q}{q}						
\newcommand{\qtavg}{\bar{\qt}}			
\newcommand{\QTavg}{\bar{\QT}}			
\newcommand{\avgq}{\bar{\qt}}			
\newcommand{\ctravg}{\bar{\ctr}}		
\newcommand{\relone}{\ctr_1}			
\newcommand{\relavgone}{\bar{\ctr}_1} 
\newcommand{\reltwo}{\ctr_2}			
\newcommand{\relavgtwo}{\bar{\ctr}_2} 
\newcommand{\SF}{\operatorname{SF}}
\newcommand{\SJ}{\operatorname{SJ}}

\newcommand{\s}{s}				

\newcommand{\vv}{\varphi}		
\newcommand{\ihr}{\lambda}		
\newcommand{\F}{F}				
\newcommand{\f}{f}				
\newcommand{\trnc}{H}			
\newcommand{\shft}{b}			

\newcommand{\M}{M}				
\newcommand{\wt}{\alpha}		
\newcommand{\wtvv}{\vv^{\wt}}	
\newcommand{\welfare}{\operatorname{welfare}}
\newcommand{\revenue}{\operatorname{revenue}}

\newcommand{\myand}{\operatorname{and}}
\newcommand{\perm}{\pi}
\newcommand{\permone}{\perm^1}
\newcommand{\permtwo}{\perm^2}

\begin{document}

\markboth{M. Sundararajan and I. Talgam-Cohen}{Refine Predictions Ad Infinitum?}

\title{Refine Predictions Ad Infinitum?}

\author{Mukund Sundararajan\thanks{Google Inc., Mountain View, CA 94043 USA.}
\and Inbal Talgam-Cohen\thanks{Computer Science Department, Stanford University, Stanford, CA 94305 USA.
The authors wish to thank Amir Najmi for suggesting the problem, Mohammad Mahdian for suggesting that Pareto optimal mechanisms are virtual value-based, and Qiqi Yan for many helpful comments.
}
}

\maketitle

\begin{abstract}
We study how standard auction objectives in sponsored search markets change with refinements in the prediction of the relevance (click-through rates) of ads.
We study mechanisms that optimize for a convex combination of efficiency and revenue. We show that the objective function of such a mechanism can only improve with refined (improved) relevance predictions, i.e., the search engine has no disincentive to perform these refinements. More interestingly, we show that under assumptions, refinements to relevance predictions can only improve the efficiency of any such mechanism. Our main technical contribution is to study how relevance refinements affect the similarity between ranking by virtual-value (revenue ranking) and ranking by value (efficiency ranking). Finally, we discuss implications of our results to the literature on signaling. 
\end{abstract}

\input{Introduction.tex}

\input{Model.tex}

\input{Pareto-Opt.tex}

\input{Refinement-Efficiency.tex}

\input{Proofs.tex}

\input{Trade-off.tex}

\input{Signaling.tex}

\section{Open Questions}

Our results rely on the multiplicative assumption, which is natural in the context of position auctions. Can they be generalized to include additional effects that refinement may have on realized values, perhaps using some linear approximation? How can the notion of \emph{flip-spread} refinement be generalized?

We have shown that for the Myerson mechanism, flip-spread refinements improves both welfare and revenue. For which other mechanisms does this desirable property hold?

\bibliographystyle{acmsmall}
\bibliography{Auctions_Bib}

\appendix
\input{Appendix.tex}

\end{document}

%% file: Introduction.tex
\section{Introduction}

Sponsored search is a multi-billion dollar market; it enables contextual advertising, and generates revenue that supports innovation in search algorithms. Besides being important, sponsored search markets are also technically interesting and have been investigated theoretically from several perspectives. For instance, auction theory (cf. ~\cite{agm,ES-10}), game theory (cf. ~\cite{var-06, EOS-07}), and bipartite matching theory (cf. ~\cite{msvv-05}). See~\cite{LPSV07} for a survey.

How do these markets operate? Market efficiency (or value maximization) is achieved by displaying relevant ads that maximize the odds of the user clicking on the impression (a shown ad), and then succesfully transacting on the advertiser's website. To do this, the search engine must acquire two very different types of information. First, it must \emph{estimate} the \emph{relevance} of an advertiser to the user's query, modeled as the probability that that advertiser's ad will receive a click when it is shown to the user. Second, it must \emph{elicit} in an incentive compatible way the \emph{value} that the advertiser has for the user's click; this quantity is determined usually by the the probability of transaction given a visit to the advertiser's site, and the profit per transaction; notice that the search engine is not privy to either quantity. The \emph{realized value} (value-per-impression) in this market is naturally modeled as a \emph{product} of the value-per-click and relevance. Indeed, it is possible that an advertiser with a low value-per-click ($\vpc_1$) and high probability of click ($\ctr_1$) would realize a higher realized value than one with a high value-per-click ($\vpc_2$) but low probability of click ($\ctr_2$), because $\vpc_1\ctr_1 > \vpc_2\ctr_2$.

The elicitation problem mentioned above is naturally modeled via auction theory (cf.~\cite{agm,ES-10}). The goal here is to maximize an auction objective such as efficiency or revenue by eliciting the value-per-click in an incentive compatible way. The estimation problem, however, is most naturally a machine learning problem~\cite{GQBH10}, that is the goal is to \emph{predict} relevance using features of the ad and the query. The relevance predictions are \emph{refined} by either improving the machine learning algorithm 
 or by adding new features. Consider an example of a refinement: Two pizza merchants, one from San Francisco and the other from the nearby city of San Jose, may appear equally relevant (both have say probability $\ctr$) for the query `pizza' emanating from an unspecified location in the Bay Area, but may have antisymmetric relevances on either side of $\ctr$ once the user's location within this region is further pinpointed.

The focus of this paper is to study how standard auction objectives, specifically efficiency, behave with relevance prediction refinements like the one above.\footnote{We study incentive-compatible mechanisms. The mechanisms used in practice, though not incentive compatible, have equilibria that are allocation- and revenue-equivalent to the corresponding incentive compatible mechanisms \cite{EOS-07, ES-10}. So we expect our results to apply to practically used mechanisms \emph{in equilibrium}.}

Conventional wisdom would suggest that refinement ought to have a positive impact on the objective for which the auction is optimal. After all, why should more information hurt? Consistent with this intuition, it has been shown very generally that refinements can only improve the efficiency of the optimally efficient mechanism, or the revenue of the revenue-optimal mechanism~\cite{FJM+12}.%
\footnote{This paper shows that relevance prediction problem is decoupled from the value-per-click elicitation problem. That is, improvements to prediction improve the revenue of the natural Myerson mechanism.}

Things begin to get more interesting when we study changes in the revenue of the optimally efficient mechanism, or the efficiency of the revenue-optimal mechanism due to refinements (indeed, the market maker may wish to optimize any combination of revenue and efficiency - see discussion in Section \ref{sec:pareto-vvbased}). For instance, the revenue of the optimally efficient mechanism can fall with refinements. It is easiest to see this in a single-slot context with two bidders (like our `pizza' example). Recall that the efficient auction allocates the slot to the advertiser with the highest realized value, and charges it the second highest realized value. The refinement suggested in the pizza example above causes relevances to become antisymmetric, and second highest realized value to drop, thereby reducing revenue.

Similarly, we can demonstrate that the efficiency of the revenue-optimal mechanism can fall with refinement -- see Examples~\ref{ex:non-distinguish} and \ref{ex:welfare-regular}. What drives these examples? Recall (or consult Section~\ref{sub:mechanisms}) that the revenue-optimal auction and the efficient auction both rank bidders (ads) by a monotone function of their bids, ignoring ones for which this function is negative, and allocate the remaining bidders in this sequence to the available ad slots. The key difference is that the two mechanisms employ different functions of the bids. In these bad examples, with refinement, the revenue-optimal ranking drifts further apart from the efficient ranking, demonstrating that the twin objectives of revenue and efficiency are not necessarily aligned in the context of refinement.

Our first, comparatively straightforward result (Section~\ref{sec:trade-off}) shows that once the search engine commits to a specific trade-off between efficiency and revenue, the resulting \emph{Pareto optimal} mechanism can only benefit from refinement. Thus, the search engine has no disincentive to perform refinement.%
\footnote{Obviously, this should not be at the cost of using features that violate user privacy.}

Our second, more technically challenging result (Section~\ref{sec:efficiency}) is to identify assumptions under which refinement improves the \emph{efficiency} of \emph{every} Pareto optimal mechanism. The first assumption is that the value-per-click distributions are i.i.d.~and satisfy the monotone hazard rate assumption, a fairly standard assumption. The second assumption is that that refinement causes the relevances of every pair of ads for every query to either reorder or grow further apart; this assumption is arguably restrictive and we discuss it in Section~\ref{sec:non-flip-spread}. We demonstrate that both assumptions are necessary via examples. Our main proof technique is to show that allocation ranking of every Pareto optimal mechanism draws closer to the efficienct ranking as refinements are performed under the assumptions mentioned earlier.

Finally, our results are also applicable to the literature on signaling (cf.~\cite{milweb-82,Boa-09,EFG+12,MS12}). Broadly, the connection can be described as follows: Realized value is a function of relevance, which can be modeled as the seller's (search engine's) private signal, and of the value-per-click, which is the bidder's (advertiser's) private signal. In our model, the seller knows that the bidder's realized value changes multiplicatively (scales) with the seller's signal. Therefore, he can elicit the bidder's signal, perform the transformation and find the bidder's realized value himself. Thus, the standard question in the signaling literature -- how much of his signal should the seller reveal to the bidders -- becomes a question of how refining predictions affects auction objectives. Our conclusion rephrased in the context of signaling is that, as long as realized values change multiplicatively with the signal and the above assumptions hold, revealing information improves both efficiency and revenue. See Section~\ref{sec:signaling} for more detail.

%% file: Model.tex
\section{The Model}

\subsection{The Sponsored Search Market}

Position auctions, also known as slot auctions, keyword auctions or sponsored search auctions, are used for selling online advertisement slots that appear next to search results \cite{agm, var-06, EOS-07}. We use a standard model for position auctions \cite{LPSV07}, and modify it slightly to allow us to discuss the effect of refining relevance predictions.

\subsubsection{Position Auctions Model}

In the standard static model for position auctions, the single-dimensional private value of every advertiser is the amount he's willing to pay per click on his ad. Advertisers' values are transformed into values per impression, by multiplying them by click probabilities, known as click-through rates. We assume that the click-through rates are separable \cite{LPSV07}, i.e., can be written as a product of the advertiser's relevance to the query, and the effect of the slot's position on the page. Both components of the click-through rates are predicted by a machine learning system (see, e.g., \cite{GQBH10}). 

We use the following notation. For a search query $\q$, the seller has $\m$ ad slots to sell to $\n$ advertisers. Advertiser $i$ has a private value $\vpc_i\in\mathbb{R}_{+}$, and a non-private click-through rate $\ctr_{(\q,i)}\sctr_{j}$ for slot $j$, where $0<\ctr_{(\q,i)}\le 1$ is the query-advertiser relevance,%
\footnote{Our results hold when $\ctr_{(\q,i)}$ is allowed to have a zero value as well; to simplify the exposition we assume it is strictly positive.%
} and $1\ge\sctr_1\ge\dots\ge\sctr_\m\ge 0$ are the decreasing slot effects. Where the query $\q$ is clear from the context we use $\ctr_i$ to denote the relevance. 

The value per impression of advertiser $i$ for slot $j$ is $\val_{i,j} = \ctr_{i}\sctr_{j}\vpc_{i}$. The search engine ranks advertisers according to the component of $\val_{i,j}$ that is not slot related:

\begin{definition}
[Realized value]
Advertiser $i$'s \emph{realized value} for a slot in query $q$ is
$$
r_{(\q,i)} = \ctr_{(\q,i)}\vpc_i
$$
\end{definition}
An important feature of the model is that the relation between the values and realized values is multiplicative. 

Note that position auctions generalize multi-item auctions, in which one or more units of a single item are sold to unit-demand bidders: by setting $\ctr_{i}=\sctr_{j}=1$ for every $i,j$ we get a standard multi-item auction with $m$ identical units.

\subsubsection{A Model of Relevance Prediction}

As described above, the relevance of each advertiser to a search query is predicted via a machine learning system. We now provide a model of the system's output, which we refer to as a \emph{prediction scheme}.

The machine learning system has access to a collection of search query and advertiser \emph{features}. Possible features include search keywords, geographic location, time, user data and search history, as well as ad text and landing page \cite{GQBH10}. Adopting the standard assumption that features are discretized, consider the set of all possible query-advertiser pairs. These pairs are partitioned into types or \emph{parts}, and for each part, the machine learning system produces an estimate of the (slot-independent) relevance probability $\ctr$.%
\footnote{We also assume that this system produces the slot-specific relevance parameters ($\s_j$'s), but we don't need to specify how this is done because our refinements apply only to the slot-independent terms ($p_i$'s).
}

For example, the following description defines a simple part: ``pizza ad, user located in the Bay Area''. By taking into account additional features, coarse parts can be divided into finer \emph{subparts}, such as ``pizza ad, user located in San Francisco''. This process is called \emph{refining} the partition. 

We denote a partition by $\QT$ and a part by $\qt$, and will often use the convention that $\QTavg$ is a coarse partition and $\qtavg$ a coarse part, whereas $\QT$ is a refined partition and $\qt$ is a subpart. Given any coarse part $\qtavg$, there is a distribution over its subparts $\{\qt\mid\qt\subseteq\qtavg\}$ arising from the underlying distributions of the features.

We can now define a prediction scheme, which is a partition and corresponding relevance predictions learned by the machine learning system:

\begin{definition}
[Prediction scheme]
A \emph{prediction scheme} is a partition $\QT$ of all query-advertiser pairs, and for every $\qt\in\QT$ a relevance prediction $\ctr_{\qt}$ for query-advertiser pairs in part $\qt$.%
\footnote{As defined, a prediction scheme is a deterministic clustering scheme \cite{GNS07}. In general a prediction scheme can also be randomized, containing a distribution over relevance predictions for every part \cite{EFG+12,MS12}. Our results hold for randomized prediction schemes as well.
}%
\end{definition}

Overloading notation we denote both a prediction scheme and its partition by $\QT$.

How is a prediction scheme $\QT$ applied in an auction for a query $\q$? For every query-advertiser pair $(\q,i)$, we say that its relevance prediction $\ctr_{(\q,i)}$ is \emph{according to $\QT$} if $\ctr_{(\q,i)}=\ctr_{\qt}$, where $\qt\in\QT$ is the part to which the pair $(\q,i)$ belongs. The prediction scheme $\QT$ is applied in an auction by setting the advertisers' relevance predictions according to $\QT$, finding the realized values and running the auction on these values.

\subsection{Prediction Refinements}
\label{sub:refine}

An important auction design issue is how finer partitions of advertiser-query pairs may affect auction objectives. As we shall show, this affects the order in which the search engine ranks advertisers who are competing for slots.%
\footnote{We assume that the choice of $T$ does not affect the \emph{bias} of the predictions. Though very fine parts could cause inaccuracies due to lack of data.} This in turn could have a significant effect on the revenue and efficiency of the auction. To study this effect we need a formal definition of prediction refinement.

\begin{definition}
[Refinement]
A prediction scheme $\QT$ is a \emph{refinement} of $\QTavg$ if its partition is a refinement of $\QTavg$'s partition, and the average relevance of all subparts of a coarse part $\qtavg$ is its relevance according to $\QTavg$, i.e., 
$$
\ctr_{\qtavg} = \mathbb{E}_{\qt\subseteq\qtavg}[\ctr_{\qt}].
$$
\end{definition}
If the subpart and its coarse counterpart are clear from context, we use $\ctr$ and $\ctravg$ to denote the corresponding relevance predictions.

We now define a natural class of refinements -- those which distinguish among the advertisers, thus enabling a better matching of advertisements to the query.

\begin{definition}
[Spread or flipped pairs]
\label{def:spread-or-flipped}
A pair $a,b$ is \emph{spread} with respect to a pair $c,d$ if 
\begin{eqnarray*}
\frac{a}{b} \ge \frac{c}{d} \ge 1 & \mbox{or} & 1 \ge \frac{c}{d} \ge \frac{a}{b};
\end{eqnarray*}
$a,b$ is \emph{flipped} with respect to $c,d$ if
\begin{eqnarray*}
\frac{a}{b} \ge 1 \ge \frac{c}{d} & \mbox{or} & \frac{c}{d} \ge 1 \ge \frac{a}{b}.
\end{eqnarray*}
\end{definition}
See Figure \ref{fig:spread-or-flipped} for an example.

\begin{figure}
\begin{centering}
\includegraphics[scale=0.3]{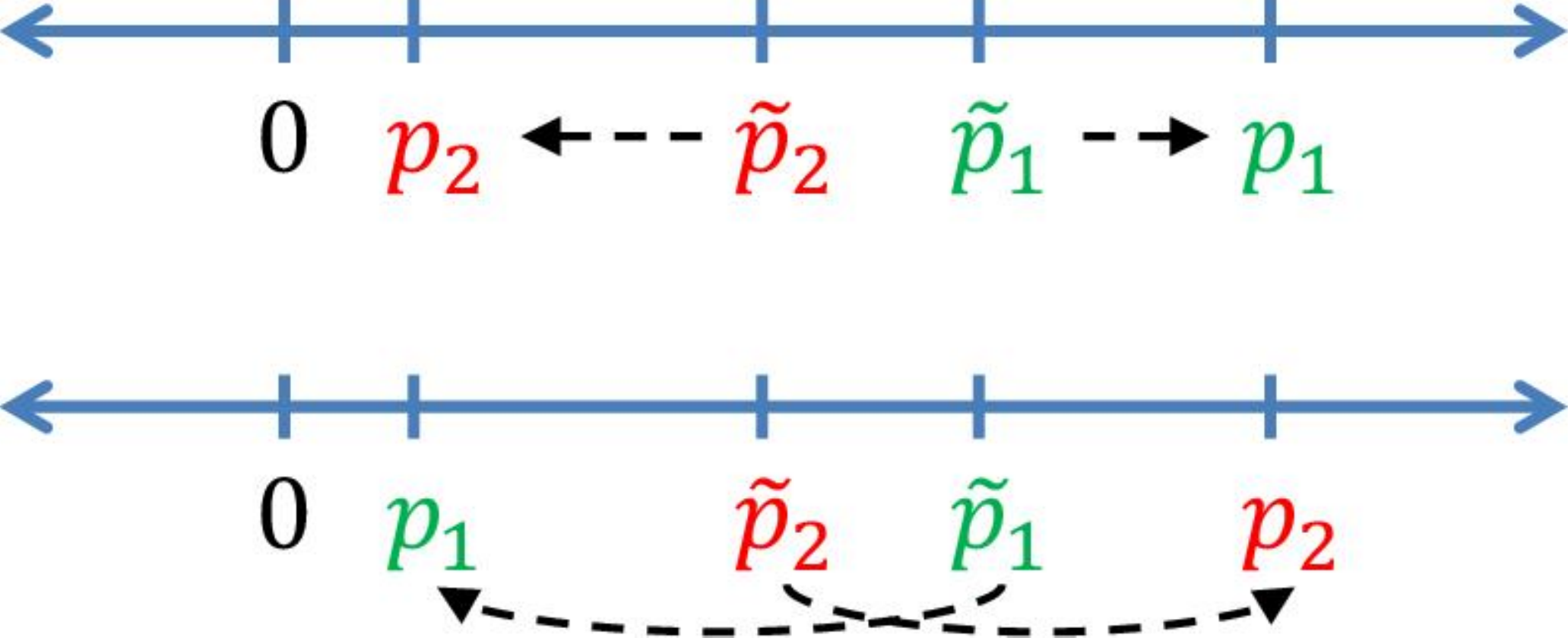}
\par\end{centering}

\label{fig:spread-or-flipped}\caption{An example of spread (above) or flipped (below) relevance pairs (see Definition \ref{def:spread-or-flipped}).}
\end{figure}

\begin{definition}
[Flip-spread refinement]
\label{def:distinguish}
A prediction scheme $\QT$ is a \emph{flip-spread refinement} of $\QTavg$ if $\QT$ is a refinement of $\QTavg$, and for every query $\q$ and pair of advertisers denoted without loss of generality 1 and 2, their relevance pair $\ctr_{(\q,1)},\ctr_{(\q,2)}$ according to $\QT$ is spread or flipped with respect to their relevance pair $\ctravg_{(\q,1)},\ctravg_{(\q,2)}$ according to $\QTavg$. 
\end{definition}
In particular, \emph{any} refinement is flip-spread for a prediction scheme $\QTavg$ that does not distinguish among advertisers: If all advertisers competing for a query $\q$ belong to the same part in $\QTavg$ and so appear equally relevant, then for every two advertisers it holds that  $\ctravg_1 / \ctravg_2 = 1$, and so any relevance pair is spread or flipped with respect to $\ctravg_1,\ctravg_2$.

\begin{example}
\label{ex:distinguish}
[Flip-spread refinement]
Consider the position auction described in the introduction, in which two pizzerias -- the first located in San Francisco (SF) and the second in San Jose (SJ) -- compete for a single advertisement slot next to search results for query `pizza'. Let $\QTavg$ be a coarse prediction scheme. Assume that both pizzerias are equally popular, and so both query-advertiser pairs belong to the same part $\qtavg$ in $\QTavg$ (described as ``pizza ad, user located in the Bay Area''), and their associated relevances are $\ctravg_1=\ctravg_2=\ctravg_{\qtavg}=0.75$. As mentioned above, this makes any refinement of $\QTavg$ flip-spread.

Now assume the search engine has access to a more precise location feature of the query $q$, indicating whether the user is in SF ($\q=\SF$) or in SJ ($\q=\SJ$), and each occurs with equal probability $1/2$. When the prediction scheme is refined by including this feature, the relevances $\ctr_1,\ctr_2$ according to the refined scheme $\QT$ behave antisymmetrically: If the query comes from a user in SF, the relevance $\ctr_{(\SF,1)}$ of the SF advertiser is 1 and the relevance $\ctr_{(\SF,2)}$ of the SJ advertiser is $0.5$. If the query comes from a user in SJ, $\ctr_{(\SJ,1)}=0.5$ and $\ctr_{(\SJ,2)}=1$. Indeed, in both cases the pair $\ctr_1,\ctr_2$ is either spread of flipped with respect to $\ctravg_1,\ctravg_2$.
\end{example}

\begin{example}
\label{ex:non-distinguish}
[Non-flip-spread refinement]
Consider again a single-slot position auction for `pizza'. Assume now that advertiser 1 is a nationwide chain of pizzerias while advertiser 2 is a local artisan pizzeria in SF. Consider a coarse prediction scheme $\QTavg$ as above, and a refinement $\QT$ where this time the refining feature indicates whether $\q=\SF$ (happens with probability $(1/4)-\delta$ for some small $\delta$) or $\q=\neg\SF$ (happens with probability $(3/4)+\delta$). The relevance of advertiser 1 does not depend on user location and his relevance predictions are $\ctravg_1=\ctr_{(\SF,1)}=\ctr_{(\neg\SF,1)}=0.8$. On the other hand, the relevance predictions for advertiser 2 are $\ctravg_2=0.1$, which is the average of $\ctr_{\SF,2}=0.4$ and $\ctr_{\neg\SF,2}=\epsilon$ (this is by setting $\delta=15\epsilon/(8-20\epsilon)$). Refinement $\QT$ is not flip-spread, since for $\q=\SF$, the advertisers relevance predictions $0.8,0.4$ are not spread or flipped with respect to the coarser predictions $0.8,0.1$.
\end{example}

%% file: Pareto-Opt.tex
\subsection{Virtual Value Based Mechanisms}
\label{sub:mechanisms}

\subsubsection{The Bayesian Setting}

We assume a Bayesian setting in which the advertisers' values are i.i.d., and are drawn from a publicly known prior distribution $\F$
with a positive smooth density $\f$. The advertisers are not symmetric since they may have different click-through rates and thus non-i.i.d.~realized values. 

Given $\F$ with density $\f$ from which value $\vpc_i$ is drawn, the \emph{inverse hazard rate} (or \emph{information rent}) of $\vpc_i$ is
$$
\ihr^{\F}(\vpc_i)=\frac{1-\F(\vpc_i)}{\f(\vpc_i)}.
$$
The \emph{virtual value} corresponding to $\vpc_i$ is $\vv^{\F}(\vpc_i)=\vpc_i-\ihr^{\F}(\vpc_i)$. A similar definition applies to the realized value: Given a distribution $G$ with density $g$ from which the realized value $\reval_i$ is drawn, the \emph{realized virtual value} of $\reval_i$ is $\vv^{G}(\reval_i)=\reval_i-\ihr^{G}(\reval_i)$.
The relation $\reval_i=\ctr_i\vpc_i$ between the value and realized value implies that
\begin{eqnarray*}
G(\reval_i)=\F(\vpc_i), & g(\reval_i)=\frac{1}{\ctr_i}\f(\vpc_i), & \vv^{G}(\reval_i)=\ctr_i\vv^{F}(\vpc_i).
\end{eqnarray*}
Note that MHR (regular) values imply MHR (regular) realized values. We omit $\F$ from the notation $\vv^{\F},\ihr^{\F}$ where
the distribution is clear from the context,%
\footnote{In particular, by $\vv(\vpc_i)$ we refer to the virtual value and by $\vv(\reval_i)$ we refer to the realized virtual value.%
} and where $\vpc_i$ is clear from the context we may use the notation $\vv_i=\vv(\vpc_i)$. 

\subsubsection{Standard Assumptions on the Prior Distribution}

A distribution $\F$ is \emph{MHR} (\emph{montone hazard rate})
if its inverse hazard rate function $\ihr(\cdot)$ is
non-increasing, and is \emph{regular} if its virtual value function
$\vv(\cdot)$ is non-decreasing. We say that values are MHR (regular) if they're drawn from an MHR (regular) distribution, and that a position auction is MHR (regular) if its advertisers' values are MHR (regular).

The assumption of MHR values is standard in the mechanism design literature (see, e.g., \cite{MM-87}). Many commonly studied distributions are MHR, including the uniform, exponential and normal distributions, and all distributions with log-concave densities \cite{Ewe09}. Every MHR distribution is regular but not vice versa; an example of a regular but non-MHR distribution is the \emph{equal revenue distribution}, defined as $\F(\vpc)=1-\frac{1}{\vpc}$. 

\subsubsection{Welfare and Revenue-Optimal Mechanisms}

Mechanisms are evaluated according to their performance in terms of
expected welfare and expected revenue, where the expectation is over
the value distribution and the query distribution. We discuss briefly the form of the efficiency-optimal and revenue-optimal auctions for sponsored search as identified by prior literature (cf.~\cite{agm, ES-10}).

The \emph{VCG} auction \cite{vic-61,cla-71,gro-73} maximizes expected
social welfare among all truthful and individually rational (IR) mechanisms by finding the most efficient allocation. Every bidder is charged his externality
-- the difference between the maximum welfare if he does not participate
in the auction and the welfare of all other bidders when he does.
In the context of position auctions, the most efficient allocation
assigns the $m$ advertisers with highest realized values $\reval_i$ to
the $m$ slots, after ordering them from high to low (we assume throughout
that ties are broken lexicographically). See~\cite{agm} for the exact form of the VCG prices in the sponsored search setting.

The \emph{revenue-optimal} mechanism~\cite{M81} maximizes expected revenue
among all truthful and IR mechanisms. To find the optimal allocation, Myerson proved the following key lemma for single-item auctions, which relates the expected revenue of an allocation rule to the (realized) \emph{virtual surplus} served by it. Recall that an allocation rule is \emph{monotone} if for every bidder $i$ and every fixed set of bids by bidders other than $i$, bidder $i$'s expected allocation (weakly) increases with its bid. We now restate Myerson's lemma in the context of position auctions.

\begin{lemma}[Myerson \cite{M81} for sponsored search auctions]
\label{lem:Myerson}
Every truthful, IR position auction has a monotone allocation rule. Moreover, its expected revenue is equal to its expected realized virtual surplus,
i.e., $\mathbb{E}_{\mathbf{\vpc}}[\sum_{i,j}\sctr_j\vv_i(\reval_i)x_{i,j}(\mathbf{\vpc})]$
where $x_{i,j}(\mathbf{\vpc})$ indicates if bidder $i$ wins slot
$j$ given value profile $\mathbf{\vpc}$.%
\footnote{Myerson also showed that there is a unique pricing rule that yields a truthful IR mechanism when coupled with a monotone allocation rule. Notice however that the form of this pricing rule is rendered unimportant because this lemma gives us a handle on revenue even without knowing the precise form of the prices.}
\end{lemma}

When values are regular, the Myerson mechanism maximizes expected revenue by assigning up to $\m$ slots to the $\le\m$ advertisers with highest \emph{non-negative} realized virtual values, in high to low order.%
\footnote{Lahaie and Pennock~\cite{LP07} study methods for optimizing revenue of sponsored search auctions given a practically motivated constraint, that we rank by functions of the form $p_i^{\alpha} v_i$. They note that Myerson's virtual value based approach does not fit this scheme.} By regularity, if an advertiser's value increases then his position can only increase, and so the allocation rule is monotone, hence truthful. 

\subsubsection{Definition of Virtual Value Based Mechanisms}

We now define a class of truthful mechanisms of which Myerson and VCG are extremal members. In Section~\ref{sec:pareto-vvbased} we will show that these mechanisms optimize convex combinations of efficiency and revenue. 

Let $v$ be an advertiser's value drawn from the distribution $F$. Recall that $\lambda(v)$ is the inverse hazard rate of $v$.

\begin{definition}
[$\wt$-virtual value]
\label{def:alpha-vv}
For $\wt\ge 0$, the $\wt$-virtual value of $\vpc$ is 
$$
\wtvv(\vpc) = \vpc - \wt\ihr(\vpc).
$$ 
\end{definition}
Observe that the $\wt$-virtual value can be rewritten as a combination of the value and virtual value: $\wtvv(\vpc) = (1-\wt)\vpc + \wt\vpc-\wt\ihr(\vpc) = (1-\wt)\vpc + \wt\vv(\vpc)$. For $\wt=0$ we get the value and for $\wt=1$ we get the virtual value.

\begin{definition}
[$\wt$-virtual value based mechanism]
For $\wt\ge 0$, the \emph{$\wt$-virtual value based mechanism} asks the advertisers to report their values $\vpc_i$, ranks them according to the \emph{realized} $\wt$-virtual values $\ctr_i\wtvv_i$, and allocates the slots to the advertisers with highest non-negative such values.
\end{definition}
The resulting mechanism is deterministic. For $\wt=0$ we get the VCG auction, and for $\wt=1$ and regular values we get the Myerson mechanism. 

\begin{lemma}[Truthfulness of $\wt$-virtual value based mechanism]
\label{lem:truthfulness}
For $0\le\wt\le 1$ and regular values, the $\wt$-virtual value based mechanism is truthful.
\end{lemma}

\begin{proof}
Since $\wtvv = (1-\wt)\vpc + \wt\vv(\vpc)$ and the value distribution is regular, $\wtvv$ is non-decreasing in $\vpc$ when $0\le\wt\le 1$, and so the allocation rule of the $\wt$-virtual value based mechanism is monotone; truthfulness follows from Lemma~\ref{lem:Myerson}.
\end{proof}

\subsection{Pareto Optimal Mechanisms}
\label{sec:pareto-vvbased}

In this section we consider mechanisms that are optimal with respect to a combination of expected welfare and revenue. Such mechanisms are termed Pareto optimal, as they lie on the Pareto frontier of welfare and revenue. Myerson and Satterthwaite \cite{MS} originally established the form of Pareto optimal mechanisms in a slightly different context; they used such mechanisms to guarantee individual rationality and incentive compatibility for bilateral trade. We study these mechanisms because we expect the search engine to use a mechanism from this family. Naively, you would think that the search engine ought only to care about its own revenue, but revenue as modeled in this paper is just `short-term' revenue. In the long run, search engines are invested in the health of their ad markets, and it is reasonable to assume that they care about more than just short-term revenue, i.e., for instance, a combination of efficiency and revenue.%
\footnote{As Likhodedov and Sandholm \cite{Likhodedov03} discuss in the context of multi-unit auctions, such mechanisms can be used to to maximize expected welfare subject to a minimum constraint on the expected revenue. More recently, Diakonikolas et al.~\cite{DPPS-12} study the computational complexity of implementing such mechanisms for single-item auctions; they show that the problem is NP-hard, and provide an FPTAS for the two bidder case.}

We now discuss the form of Pareto optimal mechanisms for sponsored search auctions and show that they are virtual value based. That is, for regular value distributions and for every fixed
trade-off $(1-\wt)\mathbb{E}[\welfare]+\wt\mathbb{E}[\revenue]$, the $\wt$-virtual value based mechanism is optimal among all truthful mechanisms.%
\footnote{This holds for irregular value distributions as well, where $\wt$-virtual value is replaced by ironed $\wt$-virtual value.
} Thus,
every point on the Pareto frontier among welfare and revenue can be
realized by a virtual value based auction. 

Recall the standard \emph{rearrangement inequality}: 
\begin{lemma}\label{lem:rearrangement}
For every two decreasing vectors $\mathbf{\reval}$ and $\mathbf{\s}$ such that $\reval_1\ge\dots\ge\reval_n$ and $\s_1\ge\dots\ge\s_n$, and for every ranking (permutation) $\perm$ of $\{1,\dots,n\}$,
\begin{equation}
\sum_i{s_i\reval_{\perm_i}} \le \sum_i{s_i\reval_i}.
\end{equation}
\end{lemma}

\begin{lemma}
[Pareto optimal mechanism]
\label{lem:Pareto-mechanism}
Consider a regular position auction and a convex combination objective $(1-\wt)\mathbb{E}[\mbox{welfare}]+\wt\mathbb{E}[\mbox{revenue}]$,
where $0\le\wt\le1$. Then the optimal mechanism for this objective
among all truthful and IR mechanisms is the $\wt$-virtual value based mechanism.
\end{lemma}

\begin{proof}
Applying Myerson's lemma (Lemma \ref{lem:Myerson}) we get that the
mechanism's objective is to maximize 
\begin{eqnarray*}
(1-\wt)\mathbb{E}_{\mathbf{v}}[\sum_{i,j}v_{i,j}x_{i,j}(\mathbf{v})] + \wt\mathbb{E}_{\mathbf{v}}[\sum_{i,j}\sctr_j\varphi(\reval_{i})x_{i,j}(\mathbf{v})] & =\\
\mathbb{E}_{\mathbf{v}}[\sum_{i,j}x_{i,j}(\mathbf{v})\cdot\left((1-\wt) \sctr_j\reval_i+\wt\sctr_j\varphi(\reval_i)\right)]
\end{eqnarray*}
Therefore, by the rearrangement inequality, the optimal allocation rule takes up to $m$ bidders with highest non-negative combinations $(1-\wt)\reval_i+\wt\varphi(\reval_i)$,
and assigns them one by one to the highest slots. This is precisely
the $\alpha$-virtual value based mechanism, which is guaranteed to be truthful
by Lemma~\ref{lem:truthfulness}.
\end{proof}

%% file: Refinement-Efficiency.tex
\section{How Does Refinement Affect Market Efficiency?}
\label{sec:efficiency}

In this section we consider the effect of refined relevance prediction on the welfare guarantees of Pareto optimal mechanisms; see Section~\ref{sec:pareto-vvbased} for a discussion of why we think search engines would use a mechanism from this class. Thus, our goal is to study how the market is affected by refined relevance predictions.  In our main technical result, we identify natural sufficient conditions under which refining the prediction improves the welfare of any Pareto optimal mechanism.

\begin{theorem}[Refined prediction improves welfare]
\label{thm:main}
Consider an i.i.d., MHR position auction and let $\M$ be a Pareto optimal mechanism. Then for every value profile $\vpcprof$, coarse prediction scheme $\QTavg$ and flip-spread refinement $\QT$, running $\M$ with $\QT$ increases welfare in comparison to running $\M$ with $\QTavg$. 
\end{theorem}

Note that this result holds entirely \emph{pointwise}, that is, it does not require averaging over the value profiles, nor does it require averaging over the subparts of $\QT$ to which the query-advertiser pairs belong to (the latter is required by Theorem \ref{thm:trade-off} on improving welfare-revenue trade-offs). 

What is the technical content of this theorem? Recall that in a single-item auction setting when values are drawn i.i.d.~from a regular (or MHR) distribution, the revenue-optimal and the efficient auction both rank bidders in the same order; of course, the revenue optimal mechanisms excludes bidders with negative virtual values by using a reserve price. In the sponsored search context, even when the value per click distributions are IID and mhr, the presence of the relevance terms cause the revenue-optimal ranking to differ from the optimal one. What we will show is that the difference between the revenue and efficiency rankings (or more precisely between the efficient and Pareto optimal rankings) diminishes with refinements so long as the conditions stated in the theorem statement above hold.

\subsection{Proof of Main Result}

In Pareto optimal auctions, bidders are ranked according to their $\wt$-virtual values $\wtvv=\vpc-\wt\ihr$. We can thus think of $\alpha*\ihr$ as a ``penalty'' on the value imposed by the seller. In expectation, the inverse hazard rate is the rent a winning bidder manages to keep to himself out of his total value when the seller applies the optimal mechanism. Thus a seller aiming to maximize expected revenue `penalizes' bidders with large rents and demotes them. The monotonicity of the relative size of the penalty in comparison to the value plays an important role in our analysis; we will show that when this quantity is falling, refinements bring the revenue optimal ordering closer to the efficient one. 

\begin{definition}[Bid penalty fraction]
For every value $\vpc\sim \F$ with inverse hazard rate $\ihr$, its \emph{penalty fraction} is $\ihr/\vpc$.
\end{definition}

For MHR distributions, the penalty fraction is non-increasing: I.e., for every two values $\vpc_1>\vpc_2$ drawn from an MHR distribution $\F$,
\begin{align}
\frac{\ihr_1}{\vpc_1} \le \frac{\ihr_2}{\vpc_2} \label{eq:monotone-fraction},
\end{align}
since the inverse hazard rates satisfy $\ihr_1 \le \ihr_2$. 

However, monotonicity of the penalty fraction does not hold for all regular distributions -- see distribution in Example \ref{ex:welfare-regular}. In a sense, mhr characterizes distributions with non-increasing penalty fractions: If the hazard rate was decreasing (i.e., the inverse hazard rates are increasing), it would be possible to shift the distribution far enough to the right so that the penalty fractions are increasing.

If, as for MHR distributions, the relative effect of the penalty diminishes as the value grows larger, then values grow further apart after transforming them into $\wt$-virtual values. Let $\vpc_1,\vpc_2$ be two i.i.d.~MHR values, and let their $\wt$-virtual values be $\wtvv_1,\wtvv_2$. Then we claim that
\begin{align}
\frac{\vpc_1}{\vpc_2} < \frac{\wtvv_1}{\wtvv_2} \implies \vpc_1 > \vpc_2 \label{eq:must-be-larger}.
\end{align}

We now prove our claim. Assume for contradiction that $\vpc_1 < \vpc_2$ and $\vpc_1\wtvv_2 < \vpc_2\wtvv_1$. Plugging in $\wtvv_i = \vpc_i-\wt\ihr_i$ we get
$$
\vpc_1\vpc_2 - \wt\vpc_1\ihr_2 < \vpc_1\vpc_2 - \wt\vpc_2\ihr_1 \iff \vpc_1\ihr_2 > \vpc_2\ihr_1,
$$
achieving a contradiction by Inequality \ref{eq:monotone-fraction}.

The next lemma is our key lemma and shows that the ranking of any Pareto optimal mechanism becomes more similar to the efficient ranking with refinement. That is, it shows that if the Pareto optimal mechanism's allocation is inefficient despite using a flip-spread refined relevance prediction, then its allocation necessarily remains inefficient when the prediction is not refined. 

\begin{lemma}
[Inefficient allocation with refined prediction]
\label{lem:streamed}
Let prediction scheme $\QT$ be a flip-spread refinement of $\QTavg$, and consider two advertisers 1 and 2 whose values are drawn from an IID regular distribution. Then
$$
\relone\vpc_1 < \reltwo\vpc_2 ~\myand~ \relone\wtvv_1 \ge \reltwo\wtvv_2 > 0 \implies \relavgone\wtvv_1 \ge \relavgtwo\wtvv_2.
$$
\end{lemma}
In words, if advertiser 1 has lower realized value than 2 but higher positive realized $\wt$-virtual value according to the refined prediction $\QT$, then this holds according to $\QTavg$ as well.

\begin{proof}
First note that $\vpc_{2}>0$, because advertiser 2's realized value is strictly higher than advertiser 1's. Combining the two inequalities on the lefthand side we get 
$$
\frac{\wtvv_1}{\wtvv_2} \ge \frac{\reltwo}{\relone} > \frac{\vpc_1}{\vpc_2}.
$$

Because the value distribution is mhr, we can applying Equation~\ref{eq:must-be-larger} to show that $\vpc_1 > \vpc_2$, and so:

\begin{align}
\frac{\wtvv_1}{\wtvv_2} \ge \frac{\reltwo}{\relone} > \frac{\vpc_1}{\vpc_2} > 1.\label{eq:first}
\end{align}

By definition of a flip-spread refinement, the pair $\relone,\reltwo$ is spread or flipped with respect to $\relavgone,\relavgtwo$. So 
\begin{align}
\frac{\reltwo}{\relone} > 1 \implies \frac{\reltwo}{\relone} \ge \frac{\relavgtwo}{\relavgone}.\label{eq:second} 
\end{align}
Equations \ref{eq:first} and \ref{eq:second} combined show that $\relavgone\wtvv_1 \ge \relavgtwo\wtvv_2$, completing the proof.
\end{proof}

%% file: Proofs.tex
We now state and prove a technical lemma that generalizes the classic rearrangement inequality (Lemma~\ref{lem:rearrangement}). Consider two orderings $\perm_1,\perm_2$ of the same ground set $(1,\dots,\n)$. We say that $\perm_1$ is \emph{more ordered} than $\perm_2$ if for every pair of elements $i>j$ which appear in order in $\perm_2$, they also appear in order in $\perm_1$, i.e., $\perm_2(i)>\perm_2(j) \implies \perm_1(i)>\perm_1(j)$.

\begin{lemma}
[Generalized rearrangement inequality]
\label{lem:rearrange}
Rename the advertisers such that their realized values are decreasing, i.e., $\reval_1\ge\dots\ge\reval_\n$. Let $\perm_1,\perm_2$ be two rankings of the advertisers such that $\perm_1$ is more ordered than $\perm_2$. Let $\mathbf{\s}$ be a vector of $\n$ decreasing slot effects $\s_1\ge\dots\ge\s_\n$. Then $\mathbf{\s}\cdot\revalvec(\perm_1) > \mathbf{\s}\cdot\revalvec(\perm_2)$.
\end{lemma}

For the proof we use the following notation:
A ranking $\perm$ is just an ordered vector of the advertisers $1,\dots,\n$. We use $i<j$ for advertisers and $x<y$ for their ranks. So $\perm_x$ is the advertiser that appears in rank $x$ in $\perm$, and $\perm(i)$ is the rank of advertiser $i$ in $\perm$. The notation $\revalvec(\perm)$ means that we are ordering the vector $\revalvec$ in the same order in which the advertisers appear in $\perm$, that is, $\reval_{\perm_1},\dots,\reval_{\perm_\n}$. So
$$
\mathbf{\s}\cdot\revalvec(\perm) = \sum_{x=1}^\n{\s_x \reval_{\perm_x}} 
= \sum_{i=1}^\n{\s_{\perm(i)} \reval_i}.
$$

\begin{proof}
We first prove the statement for two rankings $\permone,\permtwo$, which are identical except for two advertisers $i<j$ that appear consecutively in both but in flipped order. I.e, let $x=\permone(i)$ and $y=\permone(j)$, then $x=y-1$ and
$$
x=\permtwo(j)=\permtwo(i)-1=y-1
$$
(we are using here the assumption that $\permone$ is more ordered than $\permtwo$ and so $i$ must appear before $j$ in $\permone$). 

Observe that for such rankings, since the only difference is advertisers $i<j$ appearing in consecutive ranks $x<y$ in $\permone$ and in flipped ranks $y,x$ in $\permtwo$, to show that $\mathbf{\s}\cdot\revalvec(\permone) \ge \mathbf{\s}\cdot\revalvec(\permtwo)$ it's sufficient to show 
\begin{align}
\s_x\reval_i + \s_y\reval_j \ge
\s_y\reval_i + \s_x\reval_j \label{eq:rearrange}.
\end{align}
Since vectors $\mathbf{\s},\revalvec$ are decreasing, $\s_x\ge \s_y$ and $\reval_i\ge\reval_j$, and so by the standard rearrangement inequality Condition \ref{eq:rearrange} holds.

We now turn to general rankings $\permone,\permtwo$. To show that  $\mathbf{\s}\cdot\revalvec(\permone) \ge \mathbf{\s}\cdot\revalvec(\permtwo)$, we conceptually run a ``bubble sort'' on $\permtwo$ to turn it step by step into $\permone$. In every step, we compare a pair of adjacent advertisers in $\permtwo$ and swap them if their order does not match that of $\permone$. The proof is complete by noticing that since $\permone$ is more ordered to begin with, then $\permtwo$ is becoming increasingly more ordered along the process, because two advertisers are swapped only if they're in the wrong order. Therefore we can apply the above claim to get that $\mathbf{\s}\cdot\revalvec(\permtwo)$ is non-decreasing with each step, thus completing the proof. 
\end{proof}

We now show that in the more general setting of Theorem \ref{thm:main},
inefficiencies due to refinement never occur. 

\begin{proof}
[of Theorem \ref{thm:main}]
By assumption, the $\n$ advertisers have i.i.d., MHR values. We can assume that $\m\ge \n$, i.e., that there are enough slots for all the ads. This is without loss of generality since the click-through rates of the lowest slots can be set to zero, by setting their slot effects $\sctr_j=0$. Observe that in a position auction with enough slots, an assignment of the advertisers to the slots is inefficient if and only if there are two advertisers who are assigned slots and for whom the following holds: the advertiser with the lower realized value among the two gets the higher slot, and vice versa. 

Let the Pareto optimal mechanism $\M$ be the $\wt$-virtual value based mechanism which ranks the advertisers according to $\wtvv$ (Lemma \ref{lem:Pareto-mechanism}). Fix a value profile $\vpcprof$, and for every advertiser, a query-advertiser part and subpart according to $\QTavg$ and its flip-spread refinement $\QT$. This fixes for every advertiser $i$ the relevance predictions $\ctravg_i$ according to $\QTavg$ and $\ctr_i$ according to $\QT$.  

Consider the assignment chosen by $M$ when using the flip-spread refinement $\QT$ for query-advertiser relevance predictions. Assuming the chosen assignment is inefficient, without loss of generality let advertisers 1 and 2 be such that both are assigned, and advertiser 1 has a lower realized value but higher slot than advertiser 2. Formally, $\wtvv_1,\wtvv_2\ge 0$, $\relone\vpc_1 < \reltwo\vpc_2$ and $\relone\wtvv_1 \ge \reltwo\wtvv_2$. 

If $\wtvv_2 > 0$, we can now invoke Lemma \ref{lem:streamed} to get that $\relavgone\wtvv_1 \ge \relavgtwo\wtvv_2$. If $\wtvv_2 = 0$, the same inequality holds trivially. We have shown that if advertiser 1 is ranked above advertiser 2 by mechanism $M$ using the flip-spread refinement $\QT$, then $M$ will also rank 1 above 2 when using the coarse prediction scheme $\QTavg$. 

Let $\perm_1$ be the ranking of advertisers by $M$ using $\QT$, and let $\perm_2$ be the ranking of advertisers by $M$ using $\QTavg$. We have shown that $\perm_1$ is more ordered than $\perm_2$. We now apply Lemma \ref{lem:rearrange}. Note that the realized values are calculated with respect to the refined relevance prediction $\QT$, since these are the values that realize the welfare of $M$. Lemma \ref{lem:rearrange} shows that running $M$ with the flip-spread refinement $\QT$ instead of $\QTavg$ increases the welfare, thus completing the proof. 
\end{proof}

\subsection{Bad Examples and Necessity of Assumptions}

We now demonstrate via examples that the assumptions like the ones made in Theorem~\ref{thm:main} are necessary.  All examples use two advertisers, one ad position, and use the revenue-optimal mechanism.

Recall the position auction described in Example \ref{ex:distinguish}. We assume both advertisers' values are i.i.d.~and drawn from a regular distribution $\F$ with density $\f$. The realized values of the advertisers are as follows:
\begin{enumerate}
\item $\reval_{(\SF,1)}=\vpc_1$, $\reval_{(\SJ,1)}=\vpc_1/2$;
\item $\reval_{(\SJ,2)}=\vpc_2$, $\reval_{(\SF,2)}=\vpc_2/2$;
\end{enumerate}

For the analysis we fix two values $\vpc\ge \vpc'>0$ drawn from $\F$,
but do not specify which value belongs to which advertiser (the
advertisers are i.i.d.~ so each of the two possibilities occurs with probability $\frac{1}{2}$). Let $\vv,\vv'$ be the virtual values corresponding to $\vpc,\vpc'$; by regularity of $\F$, $\vv\ge\vv'$. We assume that $\vv'$ is non-negative and $\vv$ is positive, since otherwise refinement does not change the allocation of the revenue-optimal mechanism, and therefore
has no effect on welfare. We also assume without loss of generality that the user is from SJ (the other case is symmetric).

\subsubsection{\label{sub:Inefficiencies-data} When refinement reduces efficiency}

Assume $\vpc>\vpc'$ (if $\vpc=\vpc'$ then allocation using the user's location is always efficient). We show that a necessary condition for inefficiencies due to refinement to occur is 
\begin{align}
\frac{\vpc'}{\vpc}<\frac{1}{2}<\frac{\vv'}{\vv} \label{eq:cond},
\end{align}
and that if Condition \ref{eq:cond} holds, then efficiency loss of magnitude $\vpc/2-\vpc'$ occurs with probability $\frac{1}{2}$. 

Assume first that Condition \ref{eq:cond} holds. With probability $1/2$, the SF advertiser has the higher among the two values, and its realized value $\vpc/2$ is higher than the other realized value $\vpc'$ by Condition \ref{eq:cond}. If the location data is ignored, the SF advertiser with the higher value wins and the assignment is efficient. But if the location is used for refinement, the SJ advertiser's virtual value $\vv'$ is higher than $\vv/2$ (Condition \ref{eq:cond}), and so there is efficiency loss of $\vpc/2-\vpc'$.

We now show that condition \ref{eq:cond} is necessary, and there are no other cases of inefficient allocation due to refinement. Without refinement, the advertiser with the higher value wins, so an inefficiency due to refinement occurs only if the advertiser with the lower value wins and only if his realized value is lower. This yields Condition \ref{eq:cond}. Now consider the case in which the SJ advertiser has the higher value. The SF advertiser cannot win after refinement since its virtual value $\vv'/2$ is lower than $\vv$, so in this case there is no inefficiency due to refinement. 

We summarize the total expected loss in welfare from inefficiencies due to refinement:
\begin{align}
\int_{\vpc_{\min}}^{\vpc_{\max}}\int_{2\vpc'}^{\bar{\vpc}}(\vpc/2-\vpc')
f(\vpc)f(\vpc')d\vpc d\vpc' \label{eq:typeI},
\end{align}
where $[\vpc_{\min},\vpc_{\max}]$ is the range of $\F$, and $\bar{\vpc}$ is defined to be the value such that the corresponding virtual value $\vv(\bar{\vpc})$ is equal to $2\vv'$.

\subsubsection{When refinement increases efficiency}

A similar analysis to the above shows that a necessary condition for refinement to increase efficiency is:
\begin{align}
\frac{1}{2}<\min\{\frac{\vpc'}{\vpc},\frac{\vv'}{\vv}\}\label{eq:condII}.
\end{align}
When Condition \ref{eq:condII} holds, then efficiency loss of magnitude $\vpc/2-\vpc'$ occurs with probability $1/2$: If the SF advertiser gets the higher value, condition \ref{eq:condII} implies that its realized value is \emph{lower}. As above, with no refinement the SF advertiser wins and with refinement the SJ advertiser wins. So in this case lack of refinement leads to an efficiency loss of $\vpc/2-\vpc'$.

The total expected loss in welfare from inefficiencies due to coarseness is
\begin{align}
\int_{\vpc_{\min}}^{\vpc_{\max}}\int_{\vpc_{\min}}^{2\vpc'}(\vpc/2-\vpc')
f(\vpc)f(\vpc')d\vpc d\vpc' \label{eq:typeII}.
\end{align}

\subsubsection{MHR example -- refinement increases efficiency}

\begin{example}
[MHR values]
Let $F$ be the uniform distribution over $[0,1]$. There are no inefficiencies due to refinement since Condition \ref{eq:cond} never holds:
$$
\frac{\vpc'}{\vpc} < \frac{\vv'}{\vv} = \frac{2\vpc'-1}{2\vpc-1} \iff \vpc' > \vpc,
$$
in contradiction to the definition of $\vpc,\vpc'$ by which $\vpc\ge\vpc'$.
\end{example}

\subsubsection{Non-MHR example -- refinement can reduce efficiency}

\begin{example}
[Non-MHR values]
\label{ex:welfare-regular}
Let $\trnc={10}^3$ be a truncating parameter and $\shft=-1$ be a shifting parameter. Let $\F$ be a variant of the equal revenue distribution achieved by truncating its support from $[1,\infty)$ to $[1,\trnc]$, and shifting it to the right by $|\shft|$ (if $\trnc=\infty$ and $\shft=0$ we get the standard equal revenue distribution). The truncation ensures that $\F$ has finite expectation so that Myerson's theory (Lemma \ref{lem:Myerson}) applies; the reason for shifting will become apparent below. The resulting distribution over the support $[2,\trnc+1]$ is
\begin{eqnarray*}
\F(\vpc)=\frac{\trnc}{\trnc-1}(1-\frac{1}{\vpc+\shft}).
\end{eqnarray*}
\end{example}

$\F$ is regular since its virtual valuation function $\vv(\vpc)=((\vpc+\shft)^{2}/\trnc)-\shft$ is increasing; note that since $b<0$, $\vv(\vpc)$ is also strictly positive for every $\vpc$. However, $\F$ is not MHR, since its inverse hazard rate function $\ihr(\vpc)=\vpc-((\vpc+\shft)^{2}/\trnc)+\shft$ increases with $\vpc$ in the range $[2, 501]$. 

By calculating and comparing the integrals in Equations \ref{eq:typeI} and \ref{eq:typeII}, we show that the expected loss due to refinement is higher than that due to coarseness, and so in this example, the flip-spread refinement that takes into account user location hurts the expected welfare of the revenue-optimal mechanism. Details of the calculation appear in the appendix. 

\subsubsection{Non-Flip-Spread Refinements}
\label{sec:non-flip-spread}

We now consider the slightly modified setting of Example \ref{ex:non-distinguish}. Recall that in this setting, the user location sometimes made the advertisers seem more ``similar'', and as a result, in more direct competition. To take advantage of the direct competition, the revenue-optimal mechanism sometimes allocates to the advertiser with lower realized value, increasing the expected revenue but decreasing the welfare, as demonstrated below.

First, we show that \emph{pointwise} over the bids, the flip-spread assumption is necessary in a strong sense for a result such as Theorem~\ref{thm:main}. That is, if the ad relevances draw closer together without flipping, for every non-degenerate mhr distribution, there exist valuations for which efficiency of the revenue-optimal mechanism falls with refinement. To construct such an example, for any refined relevance pair $\relone > \reltwo$, you could pick two values $v_1$ and $v_2 > v_1$ such that both virtual values are positive but near-zero, the realized value after refinement of the first advertiser ($\relone \cdot v_1$) dominates the second ($\reltwo \cdot v_2$), but the realized virtual value of the second advertiser ($\reltwo \cdot \vv_2$) slightly exceeds that of the first ($\relone \cdot \vv_1$). Thus after refinement, the revenue-optimal ranking disagrees with the efficient one. Now if the relevance pair $\relone, \reltwo$ is not flip-spread with respect to the unrefined relevance pair $\relavgone, \relavgtwo$, the unrefined revenue ranking will agree with the efficient one, and thus refinement causes efficiency loss.

We now show for a specific selection of relevance parameters, refinements, and value distributions, that efficiency loss can happen \emph{in expectation} when the flip spread assumption is violated. We use the relevance probabilities from Example~\ref{ex:non-distinguish}. Assume that the advertisers' values $\vpc_1,\vpc_2$ are drawn independently from the MHR uniform distribution over range $[3,5]$. The ranges of their realized values and virtual values are as follows: Since his relevance prediction is $0.8$ whether or not the prediction scheme is refined, advertiser 1's realized value range is $[2.4,4]$ and his virtual value range is $[0.8,4]$. As for advertiser 2, there are three cases to consider:
\begin{itemize}
\item $\q=\SF$ and the refined scheme $\QT$ is applied: realized value range is $[1.2,2]$ and virtual value range is $[0.4,2]$;
\item $\q=\neg\SF$ and the refined scheme $\QT$ is applied: realized value range is $[3\epsilon,5\epsilon]$ and virtual value range is $[\epsilon,5\epsilon]$;
\item The coarse scheme $\QTavg$ is applied: realized value range is $[0.3,0.5]$ and virtual value range is $[0.1,0.5]$.
\end{itemize}

Applying the refined prediction scheme $\QT$ which uses the location data lowers the expected welfare: Observe that the realized value of the advertiser 1 is always higher, and when $\QTavg$ is applied its virtual value is always higher as well, guaranteeing an efficient allocation. But because when $\q=\SF$ relevance predictions of the advertisers become closer, the range of advertiser 2's refined virtual value overlaps that of advertiser 1, so advertiser 2 sometimes wins despite this being inefficient. 

How often would we expect such inefficienies due to non flip-spread refinements in general? Assuming the setting of parameters as above, suppose we vary the second advertiser's relevance from $0$ towards $0.8$, and plot efficiency loss against the optimally efficient outcome; see Figure \ref{fig:varying-loss}. As the figure shows, several refinements that are not flip-spread would still result in an efficiency increase. For instance, any refinement where $\relavgtwo > 0.4$ and $\reltwo$ is in the range $[\relavgtwo, 0.8]$ would cause an efficiency increase. However, when $\relavgtwo < 0.4$ and $\reltwo$ is in the range $[\relavgtwo, 0.4]$, the corresponding refinement causes an efficiency drop. Thus, we would expect that if the relevances of advertisers were initially roughly comparable (recall that for the first advertiser, $\relone = \relavgone =0.8$), any refinement ought to improve expected efficiency.

\begin{figure}
\begin{centering}
\includegraphics[scale=0.3]{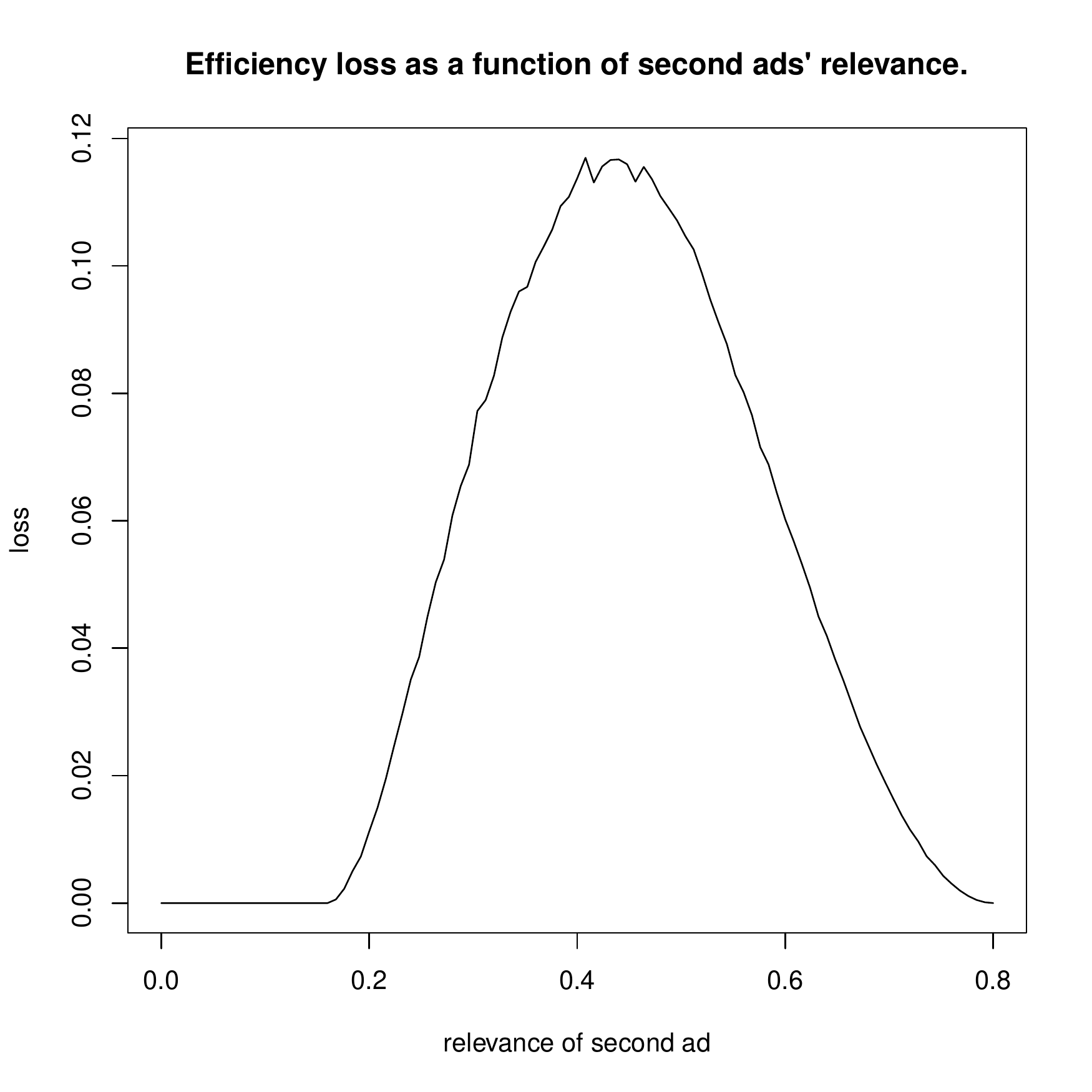}
\par\end{centering}
\label{fig:varying-loss}\caption{Efficiency loss of the revenue-optimal auction versus the efficiency-optimal auction for example in Section~\ref{sec:non-flip-spread} as a function of the relevance of the second advertiser.}
\end{figure}

\subsubsection{Non-i.i.d.~advertisers}

The previous example can be adjusted such that the resulting setting is completely equivalent, but now $\QT$ is a flip-spread refinement of $\QTavg$. This is by noticing that if advertisers are allowed to be non-i.i.d., a flip-spread refinement can make them more similar instead of more distinguished. For example, let advertiser 2's value be uniform over $[1.2,2]$ instead of $[3,5]$, and assume that finding out the user is in SF makes the advertisers' relevances flip from, say, $0.8,0.25$ to $0.8,1$. The ranges of their realized values however get closer, and the virtual value ranges overlap as above, leading to inefficiency.

%% file: Trade-off.tex
\section{Should the Search Engine Refine Predictions?}
\label{sec:trade-off}

The previous section focused on data revelation that increases the
welfare of all Pareto optimal auctions, and identified sufficient
conditions under which signaling improves welfare. In this section we investigate the search engine's incentive to perform refinements. We assume that the search engine optimizes a fixed convex combination of revenue and efficiency (see Section~\ref{sec:pareto-vvbased} for justification of this assumption), and show that refinement only improves the mixed objective. This generalizes \cite{FJM+12}, who show the same result for the revenue-optimal mechanism. 

\begin{theorem}
[Refinement improves trade-off objective]
\label{thm:trade-off}
Consider a position auction and let $M$ be a Pareto optimal mechanism. Let $T$ be a prediction scheme that is a refinement of another prediction scheme $\bar{T}$. Then running $M$ with $T$ improves the objective of $M$ in expectation in comparison to running $M$ with $\bar{T}$. 
\end{theorem}

\begin{proof}
By Lemma \ref{lem:Pareto-mechanism}, $M$ should maximize the expected realized $\wt$-virtual surplus. For every fixed search query $q$, by Lemma~\ref{lem:rearrangement}, ranking the advertisers according to their realized $\wt$-virtual values maximizes the realized $\wt$-virtual surplus, and this is achieved by using $\QT$ instead of $\QTavg$. Taking expectation over $q$ completes the proof.
\end{proof}

The implication of the above theorem is that the seller should use as refined a prediction scheme as possible.

Notice that this result is less conditional than Theorem \ref{thm:main}, that is the MHR, i.i.d.~values and flip-spread refinement assumptions are not necessary. In fact, the fact that no condition of flip-spread is imposed on the refinement means that there must be a non-trivial trade-off between efficiency and revenue. In terms of efficiency, refining ad infinitum will not always be the right thing to do. But given that the search engine has already fixed a desired trade-off among efficiency and revenue, the best thing to do in terms of this trade-off is to use the $\wt$-Pareto optimal mechansim, and refine as far as possible. 

Our main corollary follows immediately from Theorem \ref{thm:trade-off}
and Theorem \ref{thm:main}, and states that prediction refinement can simultaneously increase the welfare and the expected objective of every Pareto optimal mechanism. In particular this holds for the Myerson mechanism.

\begin{corollary}
[Refinement improves welfare and trade-off]
\label{cor:Pareto-opt}
Consider an i.i.d., MHR position auction and let $\M$ be a Pareto optimal mechanism. Then for every coarse prediction scheme $\QTavg$ and flip-spread refinement $\QT$, running $\M$ with $\QT$ increases both welfare and $M$'s objective in expectation, in comparison to running $\M$ with $\QTavg$.
\end{corollary}

%% file: Signaling.tex
\section{Relation to Signaling}
\label{sec:signaling}

\subsection{Our Result Applied to Signaling}

In our model of relevance prediction, the features of the search query that determine the relevance scores can be viewed as private information held by the seller. This information, despite being unknown to the bidders, determines their realized values for the items. Auctions with seller information appear in the seminal work of Milgrom and Weber \cite{milweb-82} and have been extensively studied in the economic literature (see, e.g., \cite{Boa-09,LM-10}), and recently also in the computer science literature, where they're sometimes called \emph{probabilistic} auctions \cite{EFG+12,MS12,FJM+12}. This reflects the fact that when the seller does not reveal his information, the bidders only know its distribution and so the realized values are stochastic. Our model is a special case in which the seller information affects the realized values \emph{multiplicatively}. We also assume that all privately held information (of the advertisers and seller) is independent.

The above previous works focus on \emph{signaling schemes}, by which the seller communicates his information to the bidders. A signaling scheme maps the seller's information to a possibly random output signal (a deterministic scheme is called a clustering scheme - see \cite{GNS07}). Based on the observed output signal, the bidders adjust their realized values according to their posterior belief about the seller's information. Overloading notation, we denote the belief when no information is released by $\avgq$ and after seeing the output signal by $\qt$. Observe that prediction refinement is mathematically equivalent to signaling, since given the features of the search query, the seller can either use a general query part $\avgq$ to predict relevance, or use more information in form of a refined subpart $\qt$. In practice, prediction refinement avoids the communication costs involved in signaling information to the advertisers,%
\footnote{This is not the case in display ads, in which some engines report the user type to the advertisers and allow them to update their bids arbitrarily.
} %
and so may be desirable when many features are involved, provided that the seller knows how to translate the reported values into realized values.

We can now state our main result in the context of signaling, under the assumption that the seller's information determines multiplicative factors by which bidder values are scaled. Let $\QTavg$ determine the multiplicative factors without revealing information and let $\QT$ replace $\QTavg$ when information is revealed. We say that the signaling scheme is flip-spread if $\QT$ is a flip-spread refinement of $\QTavg$. Then a flip-spread signaling scheme for an i.i.d., MHR position auction and Pareto optimal mechanism improves both efficiency and the mechanism's objective in expectation.

Levin and Milgrom~\cite{LM-10} discuss how fine-grained signaling can affect \emph{display} ad auctions. One of the topics they study is how revenue is impacted by `thin auctions'. Our results apply to their setting (under the assumption that signals are multiplicatively combined).

\subsection{Relation to the Linkage Principle}

The fundamental linkage principle of \cite{milweb-82} states that full revelation of seller information is revenue-optimal (see also \cite[Chapters 6 and 7]{Kri-10}). In Appendix \ref{apx:LP} we highlight the differences of our work in terms of the model, mechanism and result statement.

%% file: Appendix.tex
\section{Appendix}

\subsection{Analysis of Example \ref{ex:welfare-regular} (Efficiency Loss Due to Refinement for Regular Distribution)}

In this appendix we provide the calculations for Example \ref{ex:welfare-regular}, which demonstrates that when advertisers' i.i.d.~values are drawn from a regular but non-MHR distribution, the expected welfare of Myerson can decrease following a flip-spread refinement. 

Define $\Delta$ to be the difference between the expected efficiency loss due to refinement and the expected efficiency loss due to coarseness,
when the values are drawn from the truncated shifted equal revenue
distribution defined above in the example. In this analysis we use the notation $s$ instead of $v$. Recall that $\bar{s}$ is such that $\varphi(\bar{s})=2\varphi(s')$.

\subsubsection{Step 1}

In this step we show that 
\[
\Delta=\frac{1}{2}\int_{2}^{H+1}\left((\bar{s}-2s')F(\bar{s})+s'F(s')-\int_{s'}^{\bar{s}}F(s)ds\right)f(s')ds'.
\]

Consider the difference for fixed $s'$:

\begin{eqnarray*}
\int_{s'}^{\frac{s'}{c}}(s'-cs)f(s)ds+\int_{\frac{s'}{c}}^{\bar{s}}(s'-cs)f(s)ds & = & \int_{s'}^{\bar{s}}(s'-cs)f(s)ds\\
 & = & s'\int_{s'}^{\bar{s}}f(s)ds-c\int_{s'}^{\bar{s}}sf(s)ds.
\end{eqnarray*}

We integrate by parts:
\begin{eqnarray*}
\int_{s'}^{\bar{s}}f(s)sds & = & F(s)s\mid_{s'}^{\bar{s}}-\int_{s'}^{\bar{s}}F(s)ds\\
 & = & \bar{s}F(\bar{s})-s'F(s')-\int_{s'}^{\bar{s}}F(s)ds.
\end{eqnarray*}

Plugging in:
\begin{eqnarray*}
c\bar{s}F(\bar{s})-cs'F(s')-c\int_{s'}^{\bar{s}}F(s)ds-s'F(\bar{s})+s'F(s') & =\\
(c\bar{s}-s')F(\bar{s})+(s'-cs')F(s')-c\int_{s'}^{\bar{s}}F(s)ds & =\\
c\left((\bar{s}-\frac{s'}{c})F(\bar{s})+(\frac{s'}{c}-s')F(s')-\int_{s'}^{\bar{s}}F(s)ds\right)
\end{eqnarray*}

\subsubsection{Finding $\bar{v}$}

\begin{eqnarray*}
\frac{(\bar{s}-1)^{2}+H}{H} & = & \frac{2(s'-1)^{2}+2H}{H}\\
\bar{s} & = & 1+\sqrt{2(s'-1)^{2}+H}
\end{eqnarray*}

\subsubsection{The expressions to integrate}

We use the following:
\begin{eqnarray*}
\int_{s'}^{\bar{s}}F(s)ds & = & \frac{H}{H-1}\int_{s'}^{\bar{s}}(1-\frac{1}{s-1})ds\\
 & = & \frac{H}{H-1}\left(s-\log(s-1)\mid_{s'}^{\bar{s}}\right)\\
 & = & \frac{H}{H-1}\left(\bar{s}-\log(\bar{s}-1)-s'+\log(s'-1)\right)
\end{eqnarray*}
So the inner function before multiplying by $f(s')$ is:
\begin{eqnarray*}
\frac{1}{2}\left((\bar{s}-2s')F(\bar{s})+s'F(s')-\int_{s'}^{\bar{s}}F(s)ds\right) & =\\
\frac{H}{2\left(H-1\right)}\left((\bar{s}-2s')(1-\frac{1}{\bar{s}-1})+s'(1-\frac{1}{s'-1})-\bar{s}+\log(\bar{s}-1)+s'-\log(s'-1)\right) & =\\
\frac{H}{2\left(H-1\right)}\left(\bar{s}-2s'-\frac{\bar{s}-2s'}{\bar{s}-1}+s'-\frac{s'}{s'-1}-\bar{s}+\log(\bar{s}-1)+s'-\log(s'-1)\right) & =\\
\frac{H}{2\left(H-1\right)}\left(\log(\bar{s}-1)-\frac{\bar{s}-2s'}{\bar{s}-1}-\frac{s'}{s'-1}-\log(s'-1)\right)
\end{eqnarray*}
Multiplying by $f(s')=\frac{H}{H-1}\cdot\frac{1}{(s-1)^{2}}$ and
plugging in $\bar{s}=1+\sqrt{2(s'-1)^{2}+H}$, we get three integrals
as follows (we write them without the factor of $\frac{1}{2}\left(\frac{H}{H-1}\right)^{2}$
for simplicity):
\begin{eqnarray*}
\frac{1}{2}\int_{1-b}^{H-b}\frac{\log(2(s'-1)^{2}+H)}{(s'-1)^{2}}ds' & = & \frac{1}{2}\int_{1}^{H}\frac{\log(2x^{2}+H)}{x^{2}}dx\\
\int_{1-b}^{H-b}\frac{1+\sqrt{2(s'-1)^{2}+H}-2s'}{(s-1)^{2}\sqrt{2(s'-1)^{2}+H}}ds' & = & \int_{1}^{H}\frac{\sqrt{2x^{2}+H}-2x-1}{x^{2}\sqrt{2x^{2}+H}}dx\\
 & = & \int_{1}^{H}x^{-2}dx-\int_{1}^{H}\frac{2x+1}{x^{2}\sqrt{2x^{2}+H}}dx\\
\int_{1-b}^{H-b}\frac{s'}{\left(s'-1\right)^{3}}+\frac{\log(s'-1)}{(s-1)^{2}}ds' & = & \int_{1}^{H}\frac{x+1}{x^{3}}+\frac{\log x}{x^{2}}dx
\end{eqnarray*}

\subsubsection{Calculating the integrals}

First integral:
\begin{eqnarray*}
\frac{\sqrt{2}\arctan(\frac{\sqrt{2}x}{\sqrt{H}})}{\sqrt{H}}-\frac{\log(H+2x^{2})}{2x}\mid_{1}^{H} & =\\
\frac{\sqrt{2}\arctan(\frac{\sqrt{2}H}{\sqrt{H}})}{\sqrt{H}}-\frac{\log(H+2H^{2})}{2H}-\frac{\sqrt{2}\arctan(\frac{\sqrt{2}}{\sqrt{H}})}{\sqrt{H}}+\frac{\log(H+2)}{2} & =\\
\frac{\sqrt{8H}\left(\arctan(\sqrt{2H})-\arctan(\sqrt{2/H})\right)-\log H-\log(1+2H)+H\log(H+2)}{2H} & = & 3.51487
\end{eqnarray*}
Second integral:
\[
1-\frac{1}{H}-\left(-\frac{\sqrt{H+2x^{2}}}{Hx}-\frac{2\log(\sqrt{H(H+2x^{2})}+H)}{\sqrt{H}}+\frac{2\log(\sqrt{H}x)}{\sqrt{H}}\mid_{1}^{H}\right)=0.7298
\]
Third integral:
\begin{eqnarray*}
-\frac{4x+2x\log x+1}{2x^{2}}\mid_{1}^{H} & = & \frac{5}{2}-\frac{4H+2H\log H+1}{2H^{2}}\\
 & = & 2.4911
\end{eqnarray*}

\subsubsection{Putting everything together}

\[
\frac{1}{2}\left(\frac{H}{H-1}\right)^{2}\left(3.51487-0.7298-2.4911\right)=0.1473.
\]
The final answer is positive, and so we have shown that the inefficiencies
due to refinement surpass those due to coarseness in this example. 

\subsection{Relation to the Linkage Principle}
\label{apx:LP}
The Milgrom-Weber model is more general than ours in that realized values depend on the seller's information in an arbitrary way (and can also depend on other bidders' private values). This dependence however must be the same for all bidders. This symmetry requirement is crucial, and without it releasing seller information may actually harm revenue \cite[Chapter 8]{PR-99,Kri-10}. In contrast, an inherent feature of our model is that advertisers can be asymmetric or even antisymmetric in the way their relevance changes as more features are used for prediction. We remark that the Milgrom-Weber model also allows a certain form of correlation among private information called \emph{affiliation} (see \cite[Appendix D]{Kri-10} and \cite[FKG inequality]{AS08}). 

Milgrom and Weber do not consider direct revelation mechanisms as we do, since in their model the information of the seller and other bidders can affect realized values arbitrarily. Instead they analyze first price, second price and English auctions assuming symmetric (possibly untruthful) equilibruim bidding. Their result that using full seller information is optimal holds even for the second price auction with no reserve. This is not the case in our model, where we show that the combined objective of a Pareto optimal mechanism increases given refined prediction.